\documentclass[conference]{IEEEtran}
\usepackage{amssymb, amsmath, amsthm}
\usepackage{amsfonts,mathtools,mathrsfs,mathtools,color}
\usepackage[colorlinks]{hyperref}

\newtheorem{theorem}{Theorem}
\newtheorem{lemma}{Lemma}

\newcommand{\R}{\mathbb{R}} 
\newcommand{\part}[2]{\frac{\partial #1}{\partial #2}} 

\newcommand{\KL}{K}
 
\newcommand{\E}{\mathbb{E}}
\newcommand{\Tr}{\mathrm{Tr}}
\newcommand{\snr}{\textnormal{\textsf{snr}}}
\newcommand{\mmse}{\textnormal{\textsf{mmse}}}
\begin{document}

\title{Information and estimation \\in Fokker-Planck channels}

\author{\IEEEauthorblockN{Andre Wibisono\IEEEauthorrefmark{1},
Varun Jog\IEEEauthorrefmark{1},
and Po-Ling Loh\IEEEauthorrefmark{1}\IEEEauthorrefmark{2}}
\IEEEauthorblockA{Departments of Electrical \& Computer Engineering\IEEEauthorrefmark{1}
and Statistics\IEEEauthorrefmark{2} \\
University of Wisconsin - Madison\\
Madison, WI 53706\\
Email: aywibisono@wisc.edu, vjog@wisc.edu, loh@ece.wisc.edu}}

\maketitle

\begin{abstract}
We study the relationship between information- and estimation-theoretic quantities in time-evolving systems. We focus on the Fokker-Planck channel defined by a general stochastic differential equation, and show that the time derivatives of entropy, KL divergence, and mutual information are characterized by estimation-theoretic quantities involving an appropriate generalization of the Fisher information. Our results vastly extend De Bruijn's identity and the classical I-MMSE relation.
\end{abstract}

%
\IEEEpeerreviewmaketitle

\section{Introduction}

Information theory and statistical estimation are closely intertwined. Various identities and inequalities arise from fundamental concepts such as mutual information, Fisher information, and estimation error, and close parallels between the fields provide an avenue for devising and deriving new results. As a canonical example, the information-estimation result known as the moment-entropy inequality \cite{CovTho12} states that among all continuous random variables with a fixed variance, Gaussian random variables maximize entropy. Furthermore, Stam's inequality \cite{Sta59} states that for a fixed Fisher information, Gaussian random variables minimize entropy. The celebrated Cram\'er-Rao bound from statistics \cite{LehCas06}, which establishes a lower bound on the variance of an estimator in terms of the Fisher information, follows from the aforementioned facts.

In the past decade, significant effort has been devoted to uncovering new relationships between information-theoretic and estimation-theoretic quantities, beginning with the I-MMSE identity of Guo et al.~\cite{GuoEtAl05} for additive Gaussian noise channels.
The identity, provided in Theorem \ref{thm: immse} below, states that the derivative of the mutual information between the channel input and output with respect to the signal to noise ratio (snr) is proportional to the minimum mean square error (mmse) in estimating the input from the output. As derived in Stam~\cite{Sta59}, the result is equivalent to De Bruijn's identity (cf.\ equation \eqref{eq: debruijn}). The restatement of De Bruijn's identity in terms of the mmse spawned a host of additional information-estimation results, including an extension to non-Gaussian additive noise \cite{PalVer07}, a generalization to  the mismatched estimation setting \cite{Ver10}, and several pointwise information-estimation relations \cite{VenWei12}. Different I-MMSE type relations were also obtained for the Poisson channel \cite{AtaWei12, JiaEtAl13} and L\'{e}vy channel \cite{JiaEtAl14}.

In information theory, a ``channel" is a conditional distribution relating input symbols to output symbols. Whereas this coding-theoretic model is very useful for communication channels, however, it possesses certain drawbacks. Many real-world examples such as weather systems and financial markets are best explained as systems evolving in time according to random and deterministic influences. Although it is possible to view an evolving system as a communication channel, where the current state is the channel input and the state at a future time is the channel output, such an interpretation lacks insight about the path of the system. Nonetheless, information-theoretic ideas are still useful in characterizing the behavior of time-evolving systems. For instance, one might characterize how much information about the future is contained in the present state using quantities such as entropy, KL divergence, and mutual information. This has been discussed extensively in the climate science literature \cite{Leu90, Sch99, Kle02, Del04}.

One of the simplest and most useful ways of modeling evolving systems is via continuous-time Markov chains with a continuous state space, which may be analyzed using stochastic differential equations (SDEs). In particular, the probability distribution of such systems evolves according to a partial differential equation known as the Fokker-Planck equation. We focus on an information-theoretic analysis of time-evolving systems described by SDEs, and study the rate of change of various fundamental quantities as a function of time. We show that these rates are conveniently expressed in terms of a generalized Fisher information, so our results may be interpreted as generalizations of De Bruijn's identity for the SDE or Fokker-Planck setting. Notably, we obtain a clean identity expressing the time derivative of mutual information in terms of the mutual Fisher information, allowing us to derive new I-MMSE relations. Our results are readily specialized to specific stochastic processes, including Brownian motion, Ornstein-Uhlenbeck processes, and geometric Brownian motion.

The remainder of the paper is organized as follows: In Section~\ref{SecBackground}, we review existing results and present the family of SDEs to be analyzed.
In Section~\ref{SecMain}, we develop our main results concerning the evolution of entropy, KL divergence, and mutual information in terms of estimation-theoretic quantities.
We conclude in Section~\ref{SecDiscussion} with a discussion of open questions. Proofs are contained in the Appendix.

\section{Background and problem setup}
\label{SecBackground}

We begin by presenting formal statements of De Bruijn's identity and the I-MMSE relation, followed by a detailed characterization of the SDE framework discussed in our paper.

\subsection{DeBruijn's identity and I-MMSE}

Consider the additive Gaussian noise channel
\begin{align}\label{Eq:Gaussian}
X_t = X_0 + \sqrt{t} Z
\end{align}
where $Z \sim \mathcal{N}(0,1)$ is independent of $X_0$, and $t > 0$ is a time parameter---in this case equal to the variance of the noise---that controls how much randomness is added to the system.
As $t$ increases, we expect the output $X_t$ to be more random.
De Bruijn's identity~\cite{Sta59} confirms this intuition and asserts that
\begin{equation}\label{eq: debruijn}
\frac{d}{dt} H(X_t) = \frac{1}{2} J(X_t),
\end{equation}
where $H(X_t) \! = \! -\int_{\mathbb R}p_t(x)\log p_t(x) dx$ is the Shannon entropy with $p_t$ denoting the density of $X_t$, and
\begin{equation}\label{eq: fishinf usual}
J(X_t) = \E \left[ \left(\frac{\partial}{\partial x}\log p_t(X_t)\right)^2 \right] =  \int_{\mathbb R} \frac{p_t'(x)^2}{p_t(x)} dx > 0
\end{equation}
is the (nonparametric) Fisher information.

Another common parameterization of the channel~\eqref{Eq:Gaussian} is
\begin{equation*}
Y_ \snr = \sqrt{\snr} X + Z
\end{equation*}
where $\snr > 0$ is the signal to noise ratio and $Z \sim \mathcal{N}(0,1)$ is independent of $X$.
Guo et al.~\cite{GuoEtAl05} established the following I-MMSE relation, which states that the mutual information $I(\snr) = I(X; Y_\snr)$ increases at a rate given by the mmse, and showed their result is equivalent to De Bruijn's identity~\eqref{eq: debruijn}.
\begin{theorem}[Guo et al.~\cite{GuoEtAl05}]\label{thm: immse} 
We have
\begin{equation*}
\frac{d}{d\snr} I(\snr) = \frac{1}{2} \mmse (X \,|\, Y_\snr),
\end{equation*}
where $\mmse(X\,|\,Y_\snr) = \E[(X - \E[X \,|\, Y_\snr])^2]$ denotes the minimum mean square error for estimating $X$ from $Y_\snr$. 
\end{theorem}

In terms of the time parameterization~\eqref{Eq:Gaussian}, by setting $\snr = \frac{1}{t}$ and $X_t = \sqrt{t} Y_{1/t}$ we see that Theorem~\ref{thm: immse} is equivalent to
\begin{equation}\label{eq: immse normal}
\frac{d}{dt} I(X_0;X_t) = -\frac{1}{2t^2} \mmse (X_0 \,|\, X_t).
\end{equation}


\subsection{SDEs and Fokker-Planck equation}
\label{SecSDE}

Consider a general channel that outputs a real-valued stochastic process $(X_t)_{t \ge 0}$ following the SDE
\begin{align}\label{Eq:SDE}
dX_t = a(X_t,t) dt + \sigma(X_t,t) dW_t,
\end{align}
where $(W_t)_{t \ge 0}$ is standard Brownian motion, and $a(x,t)$ and $\sigma(x,t) > 0$ are arbitrary real-valued smooth functions. The choice $a \equiv 0$ and $\sigma \equiv 1$ generates the Gaussian channel~\eqref{Eq:Gaussian}, but the SDE framework~\eqref{Eq:SDE} is considerably more general. The drift $a(x,t) dt$ is the deterministic part of the dynamics, and the diffusion $\sigma(x,t) dW_t$ introduces randomness by incrementally adding Gaussian noise: For small $\delta > 0$, we have
\begin{align*}
X_{t+\delta} \approx X_t + a(X_t,t) \delta + \sigma(X_t,t) \sqrt{\delta} Z
\end{align*}
where $Z \sim \mathcal{N}(0,1)$ is independent of $X_t$.

A convenient way to study $X_t$ is via its density $p_t$.
When $X_t$ follows the SDE~\eqref{Eq:SDE}, the density function $p(x,t) = p_t(x)$ satisfies a partial differential equation called the {\em Fokker-Planck} equation.
This result is classical and technically requires that $a(x,t)$ and $\sigma(x,t)$ satisfy appropriate regularity and growth conditions (e.g., smoothness and Lipschitz properties~\cite{Mac92}).
\begin{lemma}\label{Lem:FP} [Proof in Appendix~\ref{AppLemFP}]
The density $p(x,t)$ of the process described by equation~\eqref{Eq:SDE} satisfies
\begin{align}\label{Eq:FP}
\part{p(x,t)}{t} \! = \! -\part{}{x}(a(x,t)p(x,t)) + \frac{1}{2} \part{^2}{x^2}(b(x,t)p(x,t)),
\end{align}
where $b(x,t) = \sigma(x,t)^2$.
\end{lemma}
Note that Lemma~\ref{Lem:FP} implies the channel~\eqref{Eq:SDE} is linear in the space of input distributions, which means a mixture of inputs produces a mixture of outputs.
Also note that when $b(x,t)$ is independent of $x$, our model falls under the PDE framework of Toranzo et al.~\cite{TorEtAl16}; however, our SDE framework~\eqref{Eq:SDE} is somewhat more general.
Consider the following examples:

\subsubsection{Brownian motion}

As discussed above, the choice $a \equiv 0$ and $\sigma \equiv 1$ produces the solution $X_t = X_0 + W_t \stackrel{d}{=} X_0 + \sqrt{t} Z \sim \mathcal{N}(X_0,t)$, which is precisely the Gaussian channel~\eqref{Eq:Gaussian}.
The Fokker-Planck equation is the classical heat equation
\begin{equation*}
\frac{\partial p(x,t)}{\partial t} = \frac{1}{2}\frac{\partial^2p(x,t)}{\partial x^2}.
\end{equation*}
Any starting point $X_0 = x_0$ gives rise to the explicit solution
\begin{align}\label{Eq:Brownian}
p(x,t) = \frac{1}{\sqrt{2\pi t}} \exp \left( -\frac{(x-x_0)^2}{2t} \right).
\end{align}

\subsubsection{Ornstein-Uhlenbeck process}

The SDE is
\begin{align*}
dX_t = -\alpha X_t dt + dW_t,
\end{align*}
with $\alpha > 0$.
This corresponds to $a(x,t) = -\alpha x$ and $\sigma(x,t) = b(x,t) \equiv 1$.
The Ornstein-Uhlenbeck process is mean-reverting and arises in stochastic modeling of interest rates and particle velocities.
The explicit solution is $X_t = e^{-\alpha t}\left(X_0 + \int_0^t e^{\alpha s} dW_s\right) \stackrel{d}{=} e^{-\alpha t} X_0 + \sqrt{\frac{1}{2\alpha}(1-e^{-2\alpha t})} Z$, so
$X_t \to \mathcal{N}(0, \frac{1}{2\alpha})$ as $t \to \infty$.
The Fokker-Planck equation is
\begin{align}\label{Eq:FPOU}
\frac{\partial p(x,t)}{\partial t} = \alpha \frac{\partial}{\partial x} (xp(x,t)) + \frac{1}{2}\frac{\partial^2p(x,t)}{\partial x^2},
\end{align}
which may also be solved explicitly \cite{Shr04}.
Note that if $\alpha \to 0$, we recover Brownian motion~\eqref{Eq:Brownian}.

\subsubsection{Geometric Brownian motion}

The SDE is
\begin{align*}
dX_t = \mu X_t dt + \sigma X_t dW_t,
\end{align*}
where $\mu \in \R$ and $\sigma > 0$.
This corresponds to $a(x,t) = \mu x$, $\sigma(x,t) = \sigma x$, and $b(x,t) = \sigma^2 x^2$.
Geometric Brownian motion is used to model asset prices in financial mathematics, notably in Black-Scholes option pricing~\cite{BlaSch73}.
The explicit solution is $X_t = X_0 \exp((\mu - \frac{\sigma^2}{2}) t + \sigma W_t)$, so $\log X_t \sim \mathcal{N}\left(\log X_0 + (\mu - \frac{\sigma^2}{2}) t, \sigma^2 t\right)$.
The Fokker-Planck equation is
\begin{align}\label{Eq:FPGBM}
\frac{\partial p(x,t)}{\partial t} = -\mu \frac{\partial}{\partial x} (xp(x,t)) + \frac{\sigma^2}{2}\frac{\partial^2}{\partial x^2}(x^2 p(x,t)),
\end{align}
which also has an explicit solution \cite{Shr04}.


\section{Main results}
\label{SecMain}

We now generalize the information-estimation relations for the Gaussian channel~\eqref{Eq:Gaussian} to the SDE channel~\eqref{Eq:SDE}.
We express the time derivatives of fundamental information-theoretic quantities---entropy, relative entropy, and mutual information---in terms of a generalized Fisher information, producing analogs of De Bruijn's identity~\eqref{eq: debruijn}.
Although entropy is not always increasing, relative entropy and mutual information decrease at a rate given by the relative and mutual Fisher information, respectively.
We further interpret the Fisher information via a generalized Bayesian Cramer-Rao lower bound and express the mutual Fisher information as the mmse of estimating a function of the input from the output, thus producing generalizations of the I-MMSE relation~\eqref{eq: immse normal}.

\subsection{Generalized Fisher information}

Let $b \colon \R \to (0,\infty)$ be a positive function.
We define the \emph{Fisher information} with respect to $b$ to be
\begin{align}\label{eq: fishinf}
\!\! J_b(X) \!=\! \E\left[b(X) \Big(\frac{\partial}{\partial x} \log p(X)\Big)^2\right]
\!=\! \int_\R \! b(x) \frac{p'(x)^2}{p(x)} dx, \!\!
\end{align}
where $p$ denotes the distribution of the random variable $X$.
When $b \equiv 1$, we obtain the usual Fisher information~\eqref{eq: fishinf usual}.
Our generalized Fisher information differs from other notions in literature in that we are still motivated by Shannon entropy.
In contrast, Toranzo et al.~\cite{TorEtAl16} consider $\phi$-Fisher information to study $\phi$-entropy, while Lutwak et al.~\cite{LutEtal05} and Bercher~\cite{Ber13} consider $q$-Fisher information to study $q$-entropy, which includes R\'enyi and Tsallis entropies. For distributions $p$ and $q$, we define the \emph{relative Fisher information} with respect to $b$:\begin{align}\label{eq: fishinfrel}
J_b(p \,\|\, q) = \int_\R p(x) b(x) \left(\part{}{x} \log \frac{p(x)}{q(x)}\right)^2 dx,
\end{align}
Alternatively, $J_b(p\,\|\,q)$ may be viewed as the Bregman divergence of $J_b$.
Recall that the KL divergence, or relative entropy, may be written as the Bregman divergence of Shannon entropy.
\begin{lemma}\label{Lem:FishBreg} [Proof in Appendix~\ref{AppFishBreg}]
\begin{equation*}
J_b(p \,\|\, q) = J_b(p) - J_b(q) - \langle \nabla J_b(q), p-q \rangle.
\end{equation*}
\end{lemma}
\noindent In particular, $J_b(p\,\|\,q) \ge 0$ shows that $J_b(p)$ is convex in $p$.

Analogous to the the mutual information, which measures the reduction in conditional entropy,
we define the \emph{mutual Fisher information} of random variables $X$ and $Y$ with respect to $b$:
\begin{align}
\label{Eq:MutFish}
J_b(X; Y) = J_b(Y \,|\, X) - J_b(Y),
\end{align}
where $J_b(Y\,|\,X) = \int_\R p_X(x) J_b(Y \,|\, X=x) \,dy$ is the conditional Fisher information. Note that $J_b(X; Y)$ is the difference between $\E[J_b(p_{Y|X}(\cdot\,|\,X))]$ and $J_b(\E[p_{Y|X}(\cdot\,|\,X)]) = J_b(Y)$. 
Since $J_b$ is convex, this implies $J_b(X; Y) \ge 0$.

\subsection{From information to estimation}
\label{SecInfoEst}

In fact, the mutual Fisher information~\eqref{Eq:MutFish} is equal to a natural generalization of the statistical Fisher information, a central estimation-theoretic quantity. Recall that the pointwise statistical Fisher information $\Phi(y) = \Phi(X\,|\,Y=y)$ of a parameterized family of distributions $\{p_{X|Y}(\,\cdot\,|\,y)\}$ is
\begin{align}\label{eq: fishinfstatusual}
\Phi(X | Y=y) \!=\! \int_\R p_{X|Y}(x|y) \Big( \frac{\partial}{\partial y} \log p_{X|Y}(x|y)\Big)^2 dx.
\end{align}
We define the statistical Fisher information with respect to $b$ as the weighted average of the pointwise Fisher information:
\begin{align}\label{eq: fishinfstat}
\Phi_b(X \,|\, Y) = \int_\R p_Y(y) \,b(y)\, \Phi(X \,|\, Y=y) \,dx.
\end{align}
The following key result
provides a bridge between information and estimation:
\begin{theorem}\label{thm: fish}  [Proof in Appendix~\ref{AppFish}]
The mutual Fisher information is equal to the statistical Fisher information:
\begin{equation*}
J_b(X ; Y) = \Phi_b(X \,|\, Y).
\end{equation*}
\end{theorem}
We now derive two theorems illustrating the intimate connections between the statistical Fisher information and quantities in estimation theory.

\subsubsection{Estimation-theoretic lower bound}

The pointwise Fisher information~\eqref{eq: fishinfstatusual} has a natural estimation interpretation via the Cramer-Rao lower bound; similarly, the statistical Fisher information~\eqref{eq: fishinfstat} provides a lower bound on the estimation error when we have a prior on the parameters.
The following result
provides a weighted version of van Trees' inequality~\cite{GilLev95}:
\begin{theorem}\label{thm: vantrees} [Proof in Appendix~\ref{AppVanTrees}]
Consider a parameterized family of distributions $\{p_{Y|X}(\,\cdot \,|\, x)\}$ with prior $p_X(x)$.
For any estimator $T(y)$ of $x$,
\begin{align*}
\E\left[\frac{1}{b(X)} (T(Y)-X)^2 \right] \ge 
\frac{1}{\Phi_b(Y|X) + J_b(X)} = 
\frac{1}{J_b(X|Y)}.
\end{align*}
\end{theorem}
Thus, the conditional Fisher information $J_b(X\,|\,Y)$ is inversely proportional to the hardness of estimating $X$ from $Y$.

\subsubsection{mmse relation}

We define the mmse of $Y$ given $X$ with respect to $b$ as
\begin{align*}
\mmse_b(Y \,|\, X) = \min_{T} \E[b(X) (T(X) - Y)^2],
\end{align*}
where the minimization is over all estimators $T(X)$.
Note that the minimizer corresponds to the conditional expectation $T(X) = \E[Y|X]$, regardless of $b$. For a parameterized family of distributions $\{p_{Y|X}(\,\cdot\,|\,x)\}$, consider the pointwise score function 
$$\varphi(x,y) = \frac{\partial}{\partial y} \log p_{Y|X}(y\,|\,x).$$
Given a prior $p_X$, by Bayes rule we can define the other conditional distribution $p_{X|Y}(x|y) \propto p_X(x) p_{Y|X}(y|x)$.
Observe that for every fixed $y$, if $X \sim p_{X|Y}(\cdot\,|y)$, then $\varphi(X,y)$ is an unbiased estimator of the nonparametric score function:
\begin{equation*}
\E[\varphi(X,y) \,|\, Y = y] = \frac{\partial}{\partial y} \log p_Y(y).
\end{equation*}

This fact leads to the following result:

\begin{theorem}\label{Thm:MMSE} [Proof in Appendix~\ref{AppMMSE}]
\begin{equation*}
J_b(X;Y) = \mmse_b(\varphi(X,Y) \,|\, Y).
\end{equation*}
\end{theorem}

\subsection{De Bruijn's identity}
\label{SecDeBruijn}

We now describe our generalizations of De Bruijn's identity.
In the statements below, we write $b_t(x) = b(x,t) = \sigma(x,t)^2$.

\subsubsection{Time derivative of entropy}

Our first result
relates the rate of change of Shannon entropy to the Fisher information:

\begin{theorem}\label{Thm:Ent} [Proof in Appendix~\ref{AppEnt}]
Let $X_t$ be the output of the SDE~\eqref{Eq:SDE}.
Then
\begin{align*}
\frac{d}{dt} H(X_t) &= \frac{1}{2} J_{b_t}(X_t) + \E\left[\part{}{x}a(X_t,t) - \frac{1}{2} \part{^2}{x^2} b(X_t,t)\right].
\end{align*}
\end{theorem}

Note that in the case of heat equation, when $a \equiv 0$ and $b \equiv 1$, this result recovers the classical De Bruijn's identity~\eqref{eq: debruijn}.
However, in the general case, the entropy does not necessarily always increase.
This may seem odd, but as the results below show, we obtain monotonicity by considering  relative entropy.

\subsubsection{Time derivative of KL divergence}

Let $\KL(p \, \|\, q) = \int_\R p(x) \log \frac{p(x)}{q(x)} dx$ denote the KL divergence.
The following result establishes that the relative entropy between any two solutions is always decreasing, with a rate given by the relative Fisher information:

\begin{theorem}\label{Thm:KL} [Proof in Appendix~\ref{AppKL}]
Let $X_t$, $Y_t$ denote the output random variables of the channel~\eqref{Eq:SDE} with distributions $p_t, q_t$. Then
\begin{align*}
\frac{d}{dt} \KL(p_t \,\|\, q_t) = -\frac{1}{2} J_{b_t}(p_t \,\|\, q_t).
\end{align*}
\end{theorem}

Thus, the KL divergence is a contracting map along any two trajectories $p_t$ and $q_t$, implying the existence of at most one fixed point of the channel; i.e., a stationary probability distribution $p_\infty$ satisfying equation~\eqref{Eq:FP}.
However, such a distribution $p_\infty$ may not always exist, as in the case of the heat equation.

\subsubsection{Time derivative of mutual information}

Recall that the mutual information satisfies $I(X;Y) = H(Y) - H(Y \,|\, X)$, where  $H(Y \,|\, X) = \int_\R p_X(x) H(Y\,|\,X=x)\, dx$ is the conditional Shannon entropy. Consider the time derivative of $I(X_0;X_t)$, where $X_t$ is the output of the channel~\eqref{Eq:SDE} with input $X_0$.
Theorem~\ref{Thm:Ent} expresses the time derivative of $H(X_t)$ in terms of the Fisher information $J_{b_t}(X_t)$; since the channel~\eqref{Eq:SDE} is linear, we also obtain a formula for the time derivative of the conditional entropy $H(X_t | X_0)$ in terms of $J_{b_t}(X_t|X_0)$. 
Recalling definition~\eqref{Eq:MutFish}, this yields the following result:
\begin{theorem}
\label{Thm:I} [Proof in Appendix~\ref{AppI}]
The output $X_t$ of the SDE channel~\eqref{Eq:SDE} with input $X_0$ satisfies
\begin{align*}
\frac{d}{dt} I(X_0; X_t) &= -\frac{1}{2} J_{b_t}(X_0;X_t).
\end{align*}
\end{theorem}

\subsection{Special cases}
\label{subsec: special}
We specialize our results to the examples in Section~\ref{SecSDE}. For more details, see Appendix~\ref{AppSpecial}.

\subsubsection{Brownian motion}

Since $b \equiv 1$, the function $J_b$ is the usual nonparametric Fisher information $J$, and Theorem \ref{Thm:Ent} yields De Bruijn's identity~\eqref{eq: debruijn}.
We calculate the conditional Fisher information using the coupling $X_t \stackrel{d}{=} X_0 + \sqrt t Z$:
\begin{align*}
J(X_t\,|\,X_0) 
= \int_{\mathbb R} p_{X_0}(x_0) \, J(x_0 + \sqrt t Z) \, dx_0 = \frac{1}{t}.
\end{align*}
The mutual Fisher information is $\Phi(X_0|X_t) = \frac{1}{t}-J(X_t) \ge 0$, which implies $J(X_t) \le \frac{1}{t}$.
Theorem~\ref{Thm:I} says that
\begin{equation*}
\frac{d}{dt} I(X_0; X_t) = \frac{1}{2} \left(J(X_t) - \frac{1}{t}\right).
\end{equation*}
Using the formula~\eqref{Eq:Brownian}, we may compute the score function:
\begin{align*}
\varphi(x_0,x_t) = \frac{\partial}{\partial x_t} \log p_{X_t|X_0}(x_t|x_0) = -\frac{1}{t}(x_t-x_0).
\end{align*}
By Theorem~\ref{Thm:MMSE}, we then obtain
\begin{align*}
\Phi(X_0\,|\,X_t) = \frac{1}{t^2} \mmse((X_t-X_0) |\, X_t) = \frac{1}{t^2} \mmse(X_0 \,|\, X_t),
\end{align*}
which, with Theorem~\ref{Thm:I}, recovers the I-MMSE identity~\eqref{eq: immse normal}.

\subsubsection{Ornstein-Uhlenbeck process}

Again, we have $J_b = J$. Applying Theorem \ref{Thm:Ent} yields
\begin{equation*}
\frac{d}{dt} H(X_t) = \frac{1}{2} J(X_t) - \alpha.
\end{equation*}
Note that if $\alpha \le 0$, entropy always increases---but if \mbox{$\alpha > 0$}, which is the regime of interest, entropy need not be monotonic. 

Using $X_t \stackrel{d}{=} e^{-\alpha_t}(X_0 + \sqrt{\frac{1}{2\alpha}(e^{2\alpha t}-1)} Z)$,
we may compute the conditional Fisher information $J(X_t \,|\, X_0) = \frac{2\alpha}{1-e^{-2\alpha t}}$
and the mutual Fisher information $$\Phi(X_0\,|\,X_t) = \frac{2\alpha}{1-e^{-2\alpha t}} - J(X_t).$$
Since $\Phi(X_0\,|\,X_t) \ge 0$, this yields the bound $J(X_t) \le \frac{2\alpha}{1-e^{-2\alpha t}}$, which monotonically decreases to $2\alpha$ as $t \to \infty$.

Using the explicit solution to equation~\eqref{Eq:FPOU}, we obtain
\begin{align*}
\varphi(x_0,x_t) = - \frac{2\alpha (x_t - e^{-\alpha t} x_0)}{1-e^{-2\alpha t}}.
\end{align*}
By Theorems~\ref{Thm:MMSE} and~\ref{Thm:I}, we then deduce
the I-MMSE relation
\begin{align*}
\frac{d}{dt} I(X_0;X_t) = -\frac{2\alpha^2 e^{-2\alpha t}}{(1-e^{-2\alpha t})^2} \mmse(X_0\,|\,X_t).
\end{align*}

\subsubsection{Geometric Brownian motion}

Since $b(x,t) = \sigma^2 x^2$, we have $J_b \neq J$.
Applying Theorem~\ref{Thm:Ent} yields
\begin{equation*}
\frac{d}{dt} H(X_t) = \frac{1}{2}J_b(X_t) + \mu - \frac{1}{2}\sigma^2.
\end{equation*}
Thus, the entropy may not increase monotonically.
Using the explicit solution to equation~\eqref{Eq:FPGBM}, we may compute
\begin{align*}
\varphi(x_0,x_t) = -\frac{1}{x_t}\left(\frac{\log x_t - \log x_0 - (\mu-\frac{\sigma^2}{2}) t}{\sigma^2 t} + 1 \right).
\end{align*}
Then by Theorem~\ref{Thm:MMSE}, we derive the I-MMSE relation
\begin{equation*}
\frac{d}{dt} I(X_0; X_t) = -\frac{1}{2\sigma^2 t^2} \mmse(\log X_0 \,|\, X_t).
\end{equation*}


\subsection{Multivariate extension}
\label{Sec:Mult}

Our results extend without difficulty to the multivariate setting where $X_t$ is a stochastic process in $\R^d$ evolving according to the SDE~\eqref{Eq:SDE}, where $a(x,t) \in \R^d$ is a drift vector, $\sigma(x,t) \in \R^{d \times d}$ is a covariance matrix, and $W_t$ is standard Brownian motion in $\R^d$.
The weight matrix is given by
\begin{align*}
b(x,t) = \sigma(x,t) \sigma(x,t)^\top,
\end{align*}
which is assumed to be uniformly positive definite.

We define the generalized Fisher information~\eqref{eq: fishinf} with respect to a positive definite matrix $b(x)$:
\begin{align*}
J_b(X) = \E\left[\|\nabla \log p(X)\|^2_{b(X)}\right] = \int_{\R^d} \frac{\|\nabla p(x)\|^2_{b(x)}}{p(x)} dx,
\end{align*}
where $\|v\|^2_{b(x)} = v^\top b(x) v = \Tr(b(x) vv^\top)$ is the Mahalanobis inner product of $v \in \R^d$.
The relative~\eqref{eq: fishinfrel} and mutual Fisher information~\eqref{Eq:MutFish} are defined similarly, and the statistical Fisher information~\eqref{eq: fishinfstat} is defined as
\begin{align*}
\Phi_b(X \,|\, Y) = \int_{\R^d} p_Y(y) \, \Tr(b(y) \Phi(X \,|\, Y=y)) \, dy,
\end{align*}
where $\Phi(X \,|\, Y=y)$ is the usual Fisher information matrix $\int_{\R^d} p_{X|Y}(x|y) (\part{}{y} \log p_{X|Y}(x|y)) (\part{}{y} \log p_{X|Y}(x|y))^\top dx$. With these definitions, our results hold unchanged (see Appendix~\ref{AppMultivariate} for more details).

\section{Discussion and future work}
\label{SecDiscussion}

We have established information-estimation identities for time-evolving systems. Our results extend the classical De Bruijn's identity, which concerns the rate of entropy growth in a Brownian motion process, to the time derivatives of entropy, KL divergence, and mutual information for processes described by general SDEs. The predictability of such systems relies on the information contained in the current state regarding future states. At a high level, the current state contains progressively less information about future states of the system; we derive the specific rates of change in information in terms of quantities arising from estimation.

Theorem \ref{Thm:I} relates the derivative of $I(X_0;X_t)$ to the statistical Fisher information, which by Theorem~\ref{thm: vantrees} is inversely proportional to the difficulty of estimating $X_0$ from $X_t$.
This difficulty should increase with time, suggesting that $I(X_0;X_t)$ may decrease in a convex manner. 
Costa \cite{Cos85} showed that the entropy of Brownian motion is a convex function, and Chen et al.~\cite{Che15} showed that the first four derivatives of entropy alternate in sign.
We conjecture a similar property for higher-order derivatives of mutual information in general SDEs.

We suspect that the relationship between Fisher information and estimation may be generalized to Bregman divergences. Our results are based on the function $f_x(y) = \frac{1}{2} b(x) y^2$, corresponding to Gaussian randomness generated via $\sqrt{b(x)} dW_t$. It may be interesting to investigate stochastic processes corresponding to general convex functions. Finally, one could explore the connection to optimal transport, which interprets the Fokker-Planck equation~\eqref{Eq:FP} as the gradient flow of relative entropy in the space of probability densities with respect to the Wasserstein metric. 
Villani \cite{Vil03} then views Theorem~\ref{Thm:KL} as an analog of De Bruijn's identity.
Note that such a viewpoint focuses on a PDE rather than SDE formulation.

\bibliographystyle{ieeetr}	
\bibliography{isit_version_app.bbl}		

\appendices

\section{Proofs of lemmas}

In this Appendix, we provide proofs of the technical lemmas stated in our paper.
Throughout, we assume that all densities and functions are smooth and rapidly decreasing, so we may differentiate under the integral sign and apply integration by parts with all boundary terms zero.
These are standard assumptions in parabolic partial differential equations, and may be ensured by constraining the functions $a$ and $\sigma$ in the channel definition~\eqref{Eq:SDE} to be smooth and Lipschitz (see Mackey~\cite{Mac92} and the references cited therein for technical details).
These assumptions hold for all our examples.

\subsection{Proof of Lemma~\ref{Lem:FP}}
\label{AppLemFP}

The Fokker-Planck equation is a standard result; we follow the outline of Mackey~\cite[$\S$11]{Mac92}.

Let $h \colon \R \to \R$ be a smooth, compactly supported function, and consider the expectation
\begin{align*}
E(t) = \E[h(X_t)] = \int_\R p(x,t) h(x) dx.
\end{align*}
We compute the time derivative $\dot E(t) = \frac{d}{dt} E(t)$ in two ways:
First, by differentiating under the integral sign,
\begin{align}\label{Eq:dot1}
\dot E(t) = \frac{d}{dt} \int_\R p(x,t) h(x) dx = \int_\R \part{p(x,t)}{t} h(x) dx.
\end{align}
Second, we compute $\dot E(t)$ as the limit of the difference $\frac{1}{\delta}(E(t+\delta) - E(\delta))$ as $\delta \to 0$.
From the channel definition~\eqref{Eq:SDE}, for small $\delta$, we have
\begin{align*}
X_{t+\delta} = X_t + a(X_t,t) \delta + \sigma(X_t,t) \sqrt{\delta} Z,
\end{align*}
where $Z \sim N(0,1)$ is independent of $X_t$.
By the second-order Taylor expansion of $h$, and ignoring terms of order smaller than $\delta$, we then have
\begin{align*}
h(X_{t+\delta}) &= h(X_t) + \sqrt{\delta} h'(X_t)\sigma(X_t,t)Z \\
&\quad+ \delta \left(h'(X_t) a(X_t,t) + \frac{1}{2} h''(X_t) \sigma(X_t,t)^2 Z^2 \right).
\end{align*}
We take the expectation of both sides.
Since $Z$ and $X_t$ are independent,
\begin{equation*}
\E[h'(X_t)\sigma(X_t,t)Z] = \E[h'(X_t)\sigma(X_t,t)]\,\E[Z] = 0,
\end{equation*}
so the middle term vanishes.
Using $\E[Z^2] = 1$, we obtain
\begin{multline*}
\E[h(X_{t+\delta})] = \E[h(X_t)] \\
+ \delta \E\Big[h'(X_t) a(X_t,t) + \frac{1}{2} h''(X_t) b(X_t,t)\Big],
\end{multline*}
which shows that
\begin{align*}
\dot E(t) = \E\Big[h'(X_t) a(X_t,t) + \frac{1}{2} h''(X_t) b(X_t,t)\Big].
\end{align*}
Using integration by parts twice, this implies that
\begin{align}
&\dot E(t) = \int_\R p(x,t) \left( h'(x) a(x,t) + \frac{1}{2} h''(x) b(x,t) \right) dx \notag \\
& \!\!\! = \int_\R h(x) \Big(\!-\part{}{x}(a(x,t)p(x,t)) + \frac{1}{2} \part{^2}{x^2} (b(x,t)p(x,t))\Big) dx,
\label{Eq:dot2}
\end{align}
where all boundary terms vanish, since $h$ has compact support.
Comparing equations~\eqref{Eq:dot1} and~\eqref{Eq:dot2}, and using the fact that $h$ is arbitrary, we conclude that
\begin{align*}
\part{p(x,t)}{t} = -\part{}{x}(a(x,t)p(x,t)) + \frac{1}{2} \part{^2}{x^2} (b(x,t)p(x,t)),
\end{align*}
which is the Fokker-Planck equation~\eqref{Eq:FP}.
\qed

Observe that the Fokker-Planck equation~\eqref{Eq:FP} conserves mass.
If $p(x,t)$ satisfies~\eqref{Eq:FP}, then
\begin{align*}
&\frac{d}{dt} \int_\R p(x,t) dx = \int_\R \frac{\partial p(x,t)}{\partial t} dx \\
&~= -\int_\R \part{}{x}(a(x,t)p(x,t)) dx + \frac{1}{2} \int_\R \part{^2}{x^2} (b(x,t)p(x,t)) dx \\
&~= 0 + 0 = 0
\end{align*}
where each integral above is zero by integration by parts.
Thus, if we start with $p(x,0)$ which is a probability density, then the solution $p(x,t)$ stays a probability density at each time $t \ge 0$.

\subsection{Proof of Lemma~\ref{Lem:FishBreg}}
\label{AppFishBreg}

We first compute the gradient $\nabla J_b(p)$ with respect to its argument $p$, which is a probability density function $p \colon \R \to \R$ representing the distribution of a random variable. Recall that $\nabla J_b(p)$ satisfies
\begin{align}
\label{EqnPie}
\langle \nabla J_b(p), \, v \rangle = \lim_{\delta \to 0} \frac{1}{\delta}(J_b(p+\delta v) - J_b(p)),
\end{align}
for all $v \colon \R \to \R$ with $\int_\R v(x) dx = 0$ (so $q+\delta v$ is still a probability density), where $\langle f,g \rangle = \int_\R f(x) g(x) dx$.

For ease of notation, we suppress dependence on $x$ in what follows, and use $'$ to denote the derivative.
We also assume that $p$ is smooth and rapidly decreasing, so all boundary terms are zero in the application of integration by parts. 
We first have the following lemma:

\begin{lemma}\label{Lem:FishBregPf}
$\nabla J_b(p) = b \frac{p'^2}{p^2} - 2\frac{(bp')'}{p}$.
\end{lemma}
\begin{proof}
Let $h \colon \R \to \R$ with $\int_\R p(x)h(x) dx = 0$, and let $\delta > 0$ be small.
Let $v = ph$, so
\begin{equation*}
\log (p+\delta v) = \log p + \log (1+\delta h) = \log p + \delta h + o(\delta).
\end{equation*}
Squaring the derivative yields
\begin{align*}
(\log (p+\delta v))'^2 = (\log p)'^2 + 2\delta h' (\log p)' + o(\delta).
\end{align*}
We now multiply the above equation by $b(p+\delta v) = bp + \delta bph$ and integrate over $x$, to obtain
\begin{align*}
& J_b(p+\delta v) = J_b(p) \! + \! \delta \int (2 bph'(\log p)' + bph (\log p)'^2) \! + \! o(\delta) \\
& \qquad = J_b(p) + \delta \int h (-2(bp(\log p)')' + bp(\log p)'^2) + o(\delta),
\end{align*}
where we have used integration by parts in the second equality. Since $v = ph$ and $h$ is arbitrary, equation~\eqref{EqnPie} implies that
\begin{align*}
\nabla J_b(p) & = \frac{-2(bp(\log p)')' + bp(\log p)'^2}{p} \\
&= \frac{-2(bp')'}{p} + b\left(\frac{p'}{p}\right)^2,
\end{align*}
as desired.
\end{proof}

Lemma~\ref{Lem:FishBreg} follows from Lemma~\ref{Lem:FishBregPf} by a direct computation.
The Bregman divergence of $J_b$ from $q$ to $p$ is 
\begin{align*}
D_{J_b}(p,q) &= J_b(p) - J_b(q) - \langle \nabla J_b(q), p-q \rangle \\
&= \int b \frac{p'^2}{p} - \int b \frac{q'^2}{q} - \int (p-q)\left(b \frac{q'^2}{q^2} - 2\frac{(bq')'}{q}\right) \\
&= \int b \frac{p'^2}{p} - \int bp \frac{q'^2}{q^2}  + 2\int \frac{p}{q} (bq')' - 2\int (bq')'.
\end{align*}
The last term above is zero (assuming $q$ is also smooth and rapidly decreasing) because it is the integral of a boundary term.
By integration by parts, the third term is
$2\int \frac{p}{q} (bq')' = -2 \int bq'(\frac{p'}{q} - \frac{pq'}{q^2})$.
Therefore, we have
\begin{align*}
D_{J_b}(p,q) &= \int b \frac{p'^2}{p} - 2\int bp'\frac{q'}{q} + \int bp \frac{q'^2}{q^2} \\
&= \int bp(\log p - \log q)'^2
= J_b(p\,\|\,q),
\end{align*}
as desired.
\qed

\section{Proofs of information-estimation theorems}

In this Appendix, we prove the theorems appearing in Section~\ref{SecInfoEst}, linking information to estimation.

\subsection{Proof of Theorem~\ref{thm: fish}}
\label{AppFish}

By Bayes' rule, we have the identity
\begin{align}\label{Eq:FishPf1}
\part{}{y}\log p_{X|Y}(x|y) = \part{}{y}\log p_{Y|X}(y|x) \!-\! \part{}{y}\log p_Y(y).
\end{align}
Observe that
\begin{align*}
\int p_{X|Y}(x|y) \part{}{y}\log p_{X|Y}(x|y) dx
= \part{}{y} \int p_{X|Y}(x|y) dx
= 0.
\end{align*}
Therefore, equation~\eqref{Eq:FishPf1} implies
\begin{align*}
\int p_{X|Y}(x|y)  \part{}{y}\log p_{Y|X}(y|x) \,dx = \part{}{y}\log p_Y(y).
\end{align*}
In particular, we have the following mmse-like relation:
\begin{align*}
&\Phi(X\,|\,Y=y) = \int p_{X|Y}(x|y) \Big(\part{}{y}\log p_{X|Y}(x|y)\Big)^2 dx \\
&=  \int p_{X|Y}(x|y) \Big(\part{}{y}\log p_{Y|X}(y|x)\Big)^2 dx 
- \Big(\part{}{y}\log p_{Y}(y)\Big)^2.
\end{align*}
Multiplying by $b(y)$ and integrating with respect to $p_Y(y)$, we obtain
\begin{align*}
\Phi(X\,|\,Y) &= \int b(y)\, p_Y(y) \,\Phi(X\,|\,Y\!=\!y)\,dy \\
&= \int\int b(y) \, p_{XY}(x,y) \Big(\part{}{y}\log p_{Y|X}(y|x)\Big)^2 dx \, dy \\
& ~~~~~~- \int b(y)\, p_Y(y) \Big(\part{}{y}\log p_{Y}(y)\Big)^2 dy \\
&= \int p_X(x) \, J_b(Y\,|\,X\!=\!x)\,dx - J_b(Y) \\
&= J_b(X;Y),
\end{align*}
as desired. \qed

\subsection{Proof of Theorem~\ref{thm: vantrees}}
\label{AppVanTrees}

We follow the outline of Gill and Levit~\cite[$\S2$]{GilLev95}.
For every $y$, we have
\begin{align*}
\int \part{}{x} p_{XY}(x,y) \, dx = p_{XY}(x,y)\Big|_{x=-\infty}^{x=+\infty} = 0.
\end{align*}
Moreover, by integration by parts,
\begin{align*}
\int x \, \part{}{x} p_{XY}(x,y) \, dx = -\int p_{XY}(x,y) dx = -p_Y(y).
\end{align*}
Therefore,
\begin{align*}
\int\int (T(y)-x) \, \part{}{x} p_{XY}(x,y) \, dx \, dy = 1,
\end{align*}
which we may also write as
\begin{align*}
\E\Big[(T(Y)-X) \, \part{}{x} \log p_{XY}(X,Y)\Big] = 1,
\end{align*}
where $(X,Y) \sim p_{XY}$.
By the Cauchy-Schwarz inequality, we then have
\begin{align}
1 & = \E\left[\frac{(T(Y)-X)}{\sqrt{b(X)}} \, \sqrt{b(X)}\part{}{x} \log p_{XY}(X,Y)\right]^2 \notag \\
&\le \E\!\left[\frac{(T(Y)-X)^2}{b(X)}\!\right]\! \E\!\left[b(X) \!\Big(\part{}{x} \log p_{XY}(X,Y)\!\Big)^2\right]\!.\label{Eq:VanTreesPf1}
\end{align}
Using the fact that
\begin{equation*}
\part{}{x} \log p_{XY}(x,y) = \part{}{x} \log p_{X}(x) + \part{}{x} \log p_{Y|X}(y|x),
\end{equation*}
we also have
\begin{align}
&\E\left[b(X) \Big(\part{}{x} \log p_{XY}(X,Y)\Big)^2\right] \notag \\
& = \E\left[b(X) \Big(\part{}{x} \log p_{X}(X) \Big)^2\right] \notag \\
&~~~+ \E\left[b(X) \Big(\part{}{x} \log p_{Y|X}(Y|X) \Big)^2\right] \notag \\
&~~~+ 2\E\left[b(X)\Big(\part{}{x} \log p_{X}(X) \Big)\Big(\part{}{x} \log p_{Y|X}(Y\,|\,X) \Big)\right] \notag \\
&= J_b(X) + \Phi_b(Y\,|\,X) \label{Eq:VanTreesPf2}
\end{align}
where the last expectation above is zero because it is equal to
\begin{align*}
\int\int &b(x) \part{}{x} p_X(x) \part{}{x} p_{Y|X}(y|x) \,dy\,dx \\
&=\int b(x) \part{}{x} p_X(x) \int \part{}{x} p_{Y|X}(y|x) \,dy \, dx 
= 0,
\end{align*}
since $\int \part{}{x} p_{Y|X}(y|x) dy = \part{}{x} \int p_{Y|X}(y|x) dy = \part{}{x} 1 = 0$.
Combining inequality~\eqref{Eq:VanTreesPf1} with equation~\eqref{Eq:VanTreesPf2} yields the desired result.
\qed

\subsection{Proof of Theorem~\ref{Thm:MMSE}}
\label{AppMMSE}

Since the minimizer in the mmse definition is the conditional expectation, we may write
\begin{align}
&\mmse_b(\varphi(X,Y)\,|\,Y) \notag \\
&~~~= \E[b(Y)(\varphi(X,Y)-\E[\varphi(X,Y)|Y])^2] \notag \\
&~~~= \E[b(Y)\varphi(X,Y)^2] - \E[b(Y)(\E[\varphi(X,Y)\,|\,Y])^2]. \label{Eq:MMSEPf1}
\end{align}
The first term in equation~\eqref{Eq:MMSEPf1} is equal to
\begin{multline*}
\int p_X(x) \int b(y) p_{Y|X}(y|x) \Big(\part{}{y} \log p_{Y|X}(y|x)\Big)^2 dy \, dx \\
= \int p_X(x) J_b(Y\,|\,X\!=\!x)\,dx
= J_b(Y\,|\,X).
\end{multline*}
On the other hand, since $\E[\varphi(X,Y)\,|\,Y] = \part{}{y} \log p_Y(Y)$, 
the second term in equation~\eqref{Eq:MMSEPf1} is equal to
\begin{align*}
\int b(y) p_Y(y) \Big( \part{}{y} \log p_Y(y)\Big)^2 dy = J_b(Y).
\end{align*}
Therefore, equation~\eqref{Eq:MMSEPf1} implies that
\begin{align*}
\mmse_b(\varphi(X,Y)\,|\,Y)
= J_b(Y\,|\,X) - J_b(Y) = J_b(X;Y),
\end{align*}
as desired.
\qed

\section{Generalizations of De Bruijn's identity}

In this Appendix, we prove the theorems in Section~\ref{SecDeBruijn}.

\subsection{Proof of Theorem~\ref{Thm:Ent}}
\label{AppEnt}

We assume that $p = p(x,t)$ is smooth and rapidly decreasing, so the boundary terms become zero when we apply integration by parts, and we may differentiate under the integral sign.

We first write
\begin{align*}
\frac{d}{dt} H(X_t)
  &= -\int \part{p}{t} \log p \, dx - \int p \part{}{t}\log p \, dx.
\end{align*}
The second integral above is equal to $\int \part{p}{t} dx = \part{}{t} \int p \, dx = \part{}{t}1 = 0$.
We then substitute the Fokker-Planck equation~\eqref{Eq:FP} to obtain
\begin{align}\label{Eq:EntCalc1}
\frac{d}{dt} H(X_t)
  &= \int \part{(ap)}{x} \log p \, dx - \frac{1}{2} \int \part{^2(bp)}{x^2} \log p \, dx.
\end{align}
By two applications of integration by parts, the first integral in equation~\eqref{Eq:EntCalc1} is equal to
\begin{align}
\int \part{(ap)}{x} \log p\, dx
  &= - \int ap \part{\log p}{x}  dx 
  = - \int a \part{p}{x} dx \notag \\
  &= \int \part{a}{x} p \, dx
  = \E\left[\part{}{x}a(X_t,t)\right]. \label{Eq:EntCalc2}
\end{align}
For the second integral in equation~\eqref{Eq:EntCalc1}, we will use the identity
\begin{align*}
\part{^2\log p}{x^2} = \frac{1}{p} \part{^2 p}{x^2} - \Big(\part{\log p}{x}\Big)^2.
\end{align*}
By integration by parts, we then have
\begin{align}
\int \part{^2(bp)}{x^2} \log p \, dx
  &= \int bp \part{^2 \log p}{x^2} dx \notag \\
  &= \int b \part{^2 p}{x^2} dx - \int bp  \Big(\part{\log p}{x}\Big)^2 dx \notag \\
  &= \int \part{^2 b}{x^2} p\, dx - \int bp  \Big(\part{\log p}{x}\Big)^2 dx \notag \\
  &= \E\left[\part{^2}{x^2}b(X_t,t)\right] - J_{b_t}(X_t).  \label{Eq:EntCalc3}
\end{align}
Combining equations~\eqref{Eq:EntCalc1},~\eqref{Eq:EntCalc2}, and~\eqref{Eq:EntCalc3} yields the desired conclusion.
\qed

\subsection{Proof of Theorem~\ref{Thm:KL}}
\label{AppKL}

For ease of notation, we write $p \equiv p_t$ and $q \equiv q_t$.
By differentiating under the integral sign and using the chain rule, we have
\begin{align}
\frac{d}{dt} K(p_t \,\|\, q_t)
  &= \frac{d}{dt} \int p \log \frac{p}{q} dx  \notag \\
  &= \int \part{p}{t} \log \frac{p}{q} dx + \!\!\int \!p \part{\log p}{t} dx - \!\!\int \!p \part{\log q}{t} dx \notag \\
  &= \int \part{p}{t} \log \frac{p}{q} dx+ 0 - \int \frac{p}{q} \part{q}{t} dx,  \label{Eq:KLDot1}
\end{align}
where the middle integral above is zero because $\int p \part{\log p}{t} dx = \int \part{p}{t} dx = \part{}{t} \int p \, dx = \part{}{t} 1 = 0$.
We apply the Fokker-Planck equation~\eqref{Eq:FP} and use integration by parts to write the first integral in equation~\eqref{Eq:KLDot1} as
\begin{align}
\int \part{p}{t} \log \frac{p}{q} dx&
  = \int \left(- \part{(ap)}{x} + \frac{1}{2} \part{^2(bp)}{x^2} \right) \log \frac{p}{q} dx \notag \\
  &= \!\!\int \!\!\left(\! ap \part{}{x} \log \frac{p}{q} + \frac{1}{2} bp \part{^2}{x^2} \log \frac{p}{q} \right)\! dx. \label{Eq:KLDot2}
\end{align}
Note that
\begin{align*}
\part{}{x} \log \frac{p}{q} = \frac{q}{p} \part{}{x} \frac{p}{q}
\end{align*}
and
\begin{align*}
\part{^2}{x^2} \log \frac{p}{q} = \frac{q}{p} \part{^2}{x^2} \frac{p}{q} - \left(\part{}{x} \log \frac{p}{q}\right)^2.
\end{align*}
Plugging these relations into equation~\eqref{Eq:KLDot2} and using integration by parts and the Fokker-Planck equation~\eqref{Eq:FP} for $q(x,t)$, we find that $\int \part{p}{t} \log \frac{p}{q} dx$ is equal to
\begin{align*}
\int &\left(aq \part{}{x} \frac{p}{q} + \frac{1}{2} bq \part{^2}{x^2} \frac{p}{q} - \frac{1}{2}bp \left(\part{}{x} \log \frac{p}{q}\right)^2\right) dx \\
  &= \int \left(-\part{(aq)}{x} + \frac{1}{2} \part{^2(bq)}{x^2}\right) \frac{p}{q} dx  - \frac{1}{2} J_{b_t}(p_t\,\|\,q_t) \\
  &= \int \part{q}{t} \frac{p}{q} dx - \frac{1}{2} J_{b_t}(p_t\,\|\,q_t).
\end{align*}
Substituting this into equation~\eqref{Eq:KLDot1}, we conclude that
\begin{align*}
\frac{d}{dt} K(p_t\,\|\,q_t) = - \frac{1}{2} J_{b_t}(p_t \,\|\, q_t),
\end{align*}
as desired.
\qed

\subsection{Proof of Theorem~\ref{Thm:I}}
\label{AppI}

Theorem~\ref{Thm:Ent} holds for any initial distribution.
In particular, if $X_0$ is a point mass at $x_0 \in \R$, we have
\begin{multline}
\label{EqnDonut}
\frac{d}{dt} H(X_t\,|\,X_0=x_0) = \frac{1}{2} J_{b_t}(X_t \,|\, X_0 = x_0) \\
+ \E\left[\part{}{x}a(X_t,t) - \frac{1}{2} \part{^2}{x^2} b(X_t,t)\,\Big|\, X_0=x_0\right].
\end{multline}
For any initial distribution $X_0$, we may average equation~\eqref{EqnDonut} over $X_0$ to obtain
\begin{align*}
\frac{d}{dt} H(X_t\,|\,X_0) &= \frac{1}{2} J_{b_t}(X_t \,|\, X_0) \\
&~~+ \E\left[\part{}{x}a(X_t,t) - \frac{1}{2} \part{^2}{x^2} b(X_t,t)\right].
\end{align*}
Combining this with Theorem~\ref{Thm:Ent}, we conclude that
\begin{align*}
\frac{d}{dt} I(X_0;X_t) &= \frac{d}{dt} (H(X_t)-H(X_t\,|\,X_0)) \\
&= \frac{1}{2} (J_{b_t}(X_t) - J_{b_t}(X_t\,|\,X_0)) \\
&= -\frac{1}{2} J_{b_t}(X_0;X_t),
\end{align*}
as desired.
\qed

\section{Details for Section~\ref{subsec: special}}
\label{AppSpecial}

For the Ornstein-Uhlenbeck process, the explicit solution to the Fokker-Planck equation~\eqref{Eq:FPOU} is given by
\begin{align*}
p(x,t) = \sqrt{\frac{\alpha}{\pi (1-e^{-2\alpha t})}} \exp\left(-\frac{\alpha(x-e^{-\alpha t} x_0)^2}{1-e^{-2\alpha t}} \right).
\end{align*}

For the geometric Brownian motion process, the explicit solution to the Fokker-Planck equation~\eqref{Eq:FPOU} is given by
\begin{align*}
p(x,t) = \frac{1}{\sqrt{2\pi \sigma^2 t}} \frac{1}{x}\exp \left(- \frac{(\log x - \log x_0 - (\mu-\frac{\sigma^2}{2}) t )^2}{2\sigma^2 t}\right).
\end{align*}

The mmse expressions in for the examples in Section~\ref{subsec: special} are then obtained through direct computations.

\section{Multivariate setting}
\label{AppMultivariate}

In the multivariate setting, our channel definition still takes the form of equation~\eqref{Eq:SDE}, where $a(x,t) \in \R^d$ is a drift vector, $\sigma(x,t) \in \R^{d \times d}$ is a covariance matrix, and the weight matrix $b(x,t) = \sigma(x,t) \sigma(x,t)^\top$ is assumed to be uniformly positive definite.
The Fokker-Planck equation~\eqref{Eq:FP} now takes the form
\begin{align*}
\frac{\partial p(x,t)}{\partial t} = -\nabla \cdot (a(x,t) p(x,t)) + \frac{1}{2} \nabla \cdot (\nabla \cdot (b(x,t) p(x,t))),
\end{align*}
where $\nabla$ is the gradient and $\nabla \cdot$ is the divergence operator.
More explicitly, we have
\begin{align*}
\frac{\partial p(x,t)}{\partial t} = &-\sum_{i=1}^d \frac{\partial}{\partial x_i} (a_i(x,t) p(x,t)) \\
&+ \frac{1}{2} \sum_{i=1}^d\sum_{j=1}^d \frac{\partial^2}{\partial x_i \partial x_j} (b_{ij}(x,t) p(x,t)).
\end{align*}
The derivation for the Fokker-Planck equation in the multivariate setting is analogous to the derivation in the univariate setting (cf.\ Mackey~\cite{Mac92}).

As explained in Section~\ref{Sec:Mult}, the generalized Fisher information is defined using the Mahalanobis inner product with respect to the weight matrix $b$, and the statistical Fisher information is defined by averaging the pointwise Fisher information with respect to the weight matrix $b$.
The weighted van Trees' inequality in Theorem~\ref{thm: vantrees} becomes
\begin{align*}
\E\left[\|T(Y)-X\|^2_{b(X)^{-1}}\right] \ge \frac{1}{\Phi_b(Y|X)+J_b(X)} = \frac{1}{J_b(X|Y)},
\end{align*}
where $\|v\|^2_{b(x)^{-1}} = v^\top b(x)^{-1} v = \mathrm{Tr}(b(x)^{-1} vv^\top)$ is the Mahalanobis inner product with respect to the inverse weight matrix $b^{-1}$.
The proof proceeds in the same way as in the univariate case (see also Gill and Levit~\cite[$\S4$]{GilLev95}).

The expression for the time derivative of entropy in Theorem~\ref{Thm:Ent} now becomes
\begin{align*}
\frac{d}{dt} H(X_t) = \frac{1}{2} J_{b_t}(X_t) \!+\! \E\left[\nabla\! \cdot a(X_t,t) - \frac{1}{2} \nabla \!\cdot (\nabla\! \cdot b(X_t,t))\right],
\end{align*}
or more explicitly,
\begin{align*}
\frac{d}{dt} H(X_t) = \frac{1}{2} J_{b_t}(X_t) &+ \sum_{i=1}^d \E\left[\part{a_i(X_t,t)}{x_i}\right] \\
&- \frac{1}{2} \sum_{i=1}^d \sum_{j=1}^d \E\left[\part{^2 b_{ij}(X_t,t)}{x_i \partial x_j}\right].
\end{align*}
The statements of the other theorems remain unchanged and the proofs remain valid.

\end{document}